\documentclass[sn-mathphys,Numbered]{sn-jnl} 
 
\usepackage{amsmath}
\usepackage{amssymb}
\usepackage{amsthm}
\usepackage[english]{babel}
\usepackage[utf8]{inputenc}
\usepackage[T1]{fontenc}
\usepackage{microtype}
\usepackage{enumerate}
\usepackage{mathrsfs}
\usepackage{mathtools}
\usepackage{graphicx}
\usepackage{float}
\usepackage{dsfont}
\usepackage[colorinlistoftodos,prependcaption]{todonotes}
\usepackage{listings}
\usepackage{booktabs}
\usepackage{mathtools}
\usepackage{adjustbox}
\usepackage{wrapfig}
\usepackage{xargs}
\usepackage{xcolor}
\usepackage{csquotes}
\usepackage{subcaption}

\usepackage{multirow}%
\usepackage{amsfonts}%
\usepackage{amsthm}%
\usepackage{mathrsfs}%
\usepackage[title]{appendix}%
\usepackage{textcomp}%
\usepackage{manyfoot}%
\usepackage{algorithm}%
\usepackage{algorithmicx}%
\usepackage{algpseudocode}%
\usepackage{listings}%
\usepackage{cancel}

\theoremstyle{plain}
\newtheorem{thm}{Theorem}[section]

\newtheorem{prop}[thm]{Proposition}

\newtheorem*{thm*}{Theorem}
\newtheorem*{lem*}{Lemma}
\newtheorem*{prop*}{Proposition}
\newtheorem*{cor*}{Corollary}

\newtheoremstyle{named}{}{}{\itshape}{}{\bfseries}{.}{.5em}{\thmnote{#3}}
\theoremstyle{named}

\theoremstyle{remark}
\newtheorem*{rem}{Remark}

\theoremstyle{definition}

\newtheorem*{defi*}{Definition}

\renewcommand{\phi}{\varphi}

\newcommand{\RR}{\mathbb{R}}
\newcommand{\NN}{\mathbb{N}}

\newcommand\ang[1]{\left\langle #1 \right\rangle}
\newcommand\abs[1]{\left| #1 \right|}
\newcommand\norm[1]{\left|\left| #1 \right|\right|}

\newcommandx{\checkk}[2][1=]{\todo[linecolor=red,backgroundcolor=red!25,bordercolor=red,#1]{check: #2}}
\newcommandx{\fix}[2][1=]{\todo[linecolor=blue,backgroundcolor=blue!25,bordercolor=blue,#1]{fix: #2}}
\newcommandx{\improve}[2][1=]{\todo[linecolor=green,backgroundcolor=green!25,bordercolor=green,#1]{improve: #2}}

\date{\today}

\author[1]{\fnm{Cristiano} \sur{Ricci}}\email{cristiano.ricci@ec.unipi.it}
\affil[1]{\orgdiv{Department of Economics and Management}, \orgname{University of Pisa}, \orgaddress{\street{Via Cosimo Ridolfi 10}, \city{Pisa}, \postcode{56124}, \state{Italy}}}
\title{A non-invariance result for the spatial AK model}

\begin{document}

\abstract{This paper deals with the positivity condition of an infinite-dimensional evolutionary equation, associated with a control problem for the optimal consumption over space. We consider a spatial growth model for capital, with production generating endogenous growth and technology of the form AK. We show that for certain initial data, even in the case of heterogeneous spatial distribution of technology and population, the solution to an auxiliary control problem that is commonly used as a candidate for the original problem is not admissible. In particular, we show that initial conditions that are non-negative, under the auxiliary optimal consumption strategy, may lead to negative capital allocations over time.}
\maketitle

\textbf{AMS}: 49K20, 93C20, 58J70

\section{Introduction}
A common feature of optimal control problems in economics is the presence of state constraints. Most models dealing with economic variables, such as human or physical capital, require the variable of interest to satisfy a positivity condition which, in the framework of optimal control, translates into a state constraint. However, it is well known that the presence of a constraint either in the state or the control makes even simple problems nontrivial, even in a finite-dimensional setting \cite{hartl1995survey}. One of the main issues one has to face in the presence of state constraints is the regularity of the Hamilton-Jacobi-Bellman (HJB) equation associated with the control problem, see \cite{cannarsa1995direct,cannarsa1991boundary,capuzzo1990hamilton,faggian2008hamilton,soner1986optimal}. 
 
To overcome this difficulty, one possible approach is to consider a \emph{relaxed state constraint}. 
This amounts to considering an admissible set for the state variable that strictly contains the original one, thus allowing for a larger class of controls where one has to find the optimum. In practice, one usually considers a relaxed constraint which is more ``natural'' from the mathematical point of view, so that the corresponding problem admits an explicit solution. This often simplifies the problem so that it becomes possible to solve the corresponding HJB equation to compute the value function, identifying the optimal control. 
This approach has been successfully applied in various economic contexts, see for example   \cite{bambi2017generically,bambi2012optimal,boucekkine2013spatial,boucekkine2021control,boucekkine2019growth,fabbri2008solving}. Then one has to verify, ex-post, that the optimal auxiliary trajectory satisfies the original state constraint, proving that the actual control problem has been solved. However, it turns out that verifying the admissibility condition of the solution of the auxiliary problem for the original one is a nontrivial task. 

In this paper we focus on the verification of the positivity condition for the spatial growth model originally introduced by \cite{boucekkine2013spatial}, and extended to the case of heterogeneous coefficients in \cite{boucekkine2019growth}. We thus consider a model of spatial accumulation of capital (capital intended in a broad sense) $K(t,\theta)$ for $\theta \in S^{1}$ where the production function is assumed to be of the form $AK$. 
The dynamic of capital distribution is assumed to satisfy the following PDE:
\begin{equation}\label{eq:optimalProblemIntro}
\begin{cases}
\frac{\partial K}{\partial t}(t,\theta) = \sigma \frac{\partial^{2} K}{\partial \theta^{2}}(t,\theta) + A(\theta)K(t,\theta) - \eta(\theta) C(t,\theta),\\
K(0,\theta) = K_{0}(\theta).
\end{cases}
\end{equation}
We then consider a planner problem, where the planner has to choose the optimal consumption strategy, $C(t,\theta)$, maximizing an intertemporal Benthamite social welfare function for every (exogenous) spatial distribution of population $\eta(\theta)$ and technological level $A(\theta)$:
\begin{equation*}
\int_{0}^{+\infty} e^{-\rho t}\left( \int_{S^{1}} \frac{C(t,\theta)^{1-\gamma}}{1-\gamma}\eta(\theta)^{q}\,d\theta \right)\,dt.
\end{equation*}
The problem is complemented by the state constraint that imposes that capital should remain positive at all times, i.e. we consider the admissible set for the controls as:
\begin{equation}\label{eq:admissibleSetIntro}
\left\{ C \in L^{1}_{loc}(\RR^{+}; E^{+})\,|\, K(t,\theta) \geq 0 \forall (t,\theta) \in \RR^{+}\times S^{1} \right\},
\end{equation}
where $\RR^{+} := [0,+\infty)$ and $E^{+} := \{x \in E\,|\, x(\cdot) \geq 0\}$. Classically, the problem is studied, e.g. in \cite{banachLattice2021},  with a strict state constraint, imposing a strictly positive level of capital $K$. In this work, however, we will consider the case where the state constraint only requires capital to be positive, i.e. $K \geq 0$.
In this setting, there is no state intervention in capital reallocation, meaning that capital can flow from one location to another only by the effect of diffusion. The social planner can only optimize the consumption in every point in space $\theta \in S^{1}$, with the constraint that capital should remain positive at every point. 

Notice that the admissible set defined in Equation \eqref{eq:admissibleSetIntro} imposes a local constraint on the level of capital at every point in space. Differently, the relaxed state constraint that has been considered in the literature imposes the level of capital should remain positive in an aggregate sense: 
\begin{equation}\label{eq:admissibleSetAuxIntro}
\left\{ C \in L^{1}_{loc}(\RR^{+}; E^{+})\,|\, \ang{K(t,\cdot),b_{0}^{*}} \geq 0 \forall t \in \RR^{+} \right\},
\end{equation} 
where $b_{0}^{*}$ is a suitable eigenfunction of the Sturm-Liouville operator and $\ang{f,g} := \int_{S^{1}}f(\theta)g(\theta)\,d\theta$ denotes the usual inner product of $L^{2}$. This leads to an optimal consumption strategy on the admissible set defined in Equation \eqref{eq:admissibleSetAuxIntro} given by (see \cite[Theorem 5.7, Equation (5.14)]{banachLattice2021}): 
\begin{equation}\label{eq:optimalConsumptionAuxIntro}
C^{*}(t,\theta) := \eta(\theta)^{\frac{q-1}{\gamma}} \left(\alpha b_{0}(\theta)\right)^{-\frac{1}{\gamma}}  \int_{S^{1}}K(t,\xi)b_{0}(\xi)\,d\xi.
\end{equation}
Notice that the latter optimal consumption policy is not heterogeneous with respect to the level of capital $K$: the heterogeneity across space only depends on the population density $\eta(\cdot)$ and other exogenous factors. This, in particular, implies that the optimal consumption would be equal across space, in the absence of other exogenous sources of heterogeneity.
In the original works \cite{boucekkine2013spatial,boucekkine2019growth}, the authors relied only on a numerical verification of the admissibility of the solution of the auxiliary problem for the original one. Moreover, they also showed the admissibility of the optimal path of an auxiliary problem when the initial condition is proportional to the steady state of the problem. A remarkable step forward is provided in \cite[Corollary 5.14]{banachLattice2021} that is, up to our knowledge, the first result where the authors rigorously proved the admissibility of the auxiliary optimal path for a set of initial conditions that is sufficiently rich. In particular, the authors show that when the initial condition is close (not only proportional) to the steady state of Equation \eqref{eq:optimalProblemIntro}, then the corresponding optimal path is admissible for the original problem. Another fundamental contribution is given by \cite[Theorem 5.7]{calvia2023optimal}, where the authors study the problem in finite dimension, and show that there exist a certain configuration of parameters for which all strictly positive solution will remain positive at all times. However, the technique to show such a result is based on the properties of Metzler matrices, which ensures positivity for linear, but only finite-dimensional, dynamical systems. 

In this work, we investigate the positivity of solutions for Equation \eqref{eq:optimalProblemIntro}, with consumption given by Equation \eqref{eq:optimalConsumptionAuxIntro}. 
Most of the results concerning invariance for infinite-dimensional problems provide sufficient conditions for the invariance, see among others \cite{martin1973differential,pavel1977invariant,pavel1983semilinear,shuzhong1989viability}. 
The only results concerning necessary conditions are present in \cite[Theorem 8.1.1]{carja2007viability}, and more recently in \cite[Theorem 4.2]{cannarsa2018invariance}. In this manuscript, we are interested in showing that if the initial condition of Equation \eqref{eq:optimalProblemIntro} lies in the positive cone, that is we assume that $K_{0}(\theta) \geq 0$,  then solutions may become negative over time. In particular in Theorems \ref{teo:noninvarianceL2} and \ref{teo:noninvarianceC} we show, in an $L^{2}(S^{1})$ and $C(S^{1})$ setting respectively, that for any given configuration of parameters, there exists a non-negative initial condition $K_0(\theta)$ which leads to negative solutions for Equations \eqref{eq:optimalProblemIntro}-\eqref{eq:optimalConsumptionAuxIntro}. 
This highlights that the solution to the auxiliary problem is not admissible to the original one, and suggests that under certain initial configurations, the optimal consumption strategy must depend on the level capital $K(t,\theta)$ in each location, and not only by its aggregate level as prescribed by Equation \eqref{eq:optimalConsumptionAuxIntro}.  The invariance condition introduced in \cite{cannarsa2018invariance} is the main tool that we will use since the necessity of such a condition allows us to show our converse results. 

This paper is organized as follows: in Section \ref{sec:optimalcontrol} we introduce Equation \eqref{eq:optimalProblemIntro} as the solution to the auxiliary control problem. In Section \ref{sec:noninvariance} we state and prove our main results, that is the non-invariance of the positive cone in the Banach spaces $L^{2}(S^{1})$ and $C(S^{1})$. In Section \ref{sec:numerical} we investigate numerically the problem, and finally in Section \ref{sec:conclusion} we conclude by discussing possible extensions to our approach to the cases of the strictly positive cone. 

\section{The Optimal Control Problem}\label{sec:optimalcontrol}
In this section, we motivate the analysis of the positivity of solutions for Equation \eqref{eq:optimalProblemIntro}. 

\subsection{Preliminaries}
We start by introducing all the notations used in this manuscript. \\
Consider the set $S^{1}$ as the interval $[0,1]$ with identified endpoints: for every $x\in\RR$ we consider the quantity $x\pmod{1}\in[0,1)$ (that is the fractional part of $x$); positions in $[0,1)$ are in one-to-one correspondence with positions in $S^1$. 
To measure distances consistently on the torus we define, for every $x,y\in[0,1)$, the distance
\begin{equation*}
  d_{S^1}(x,y):=\min\left\{(x-y)\pmod{1},(y-x)\pmod{1}\right\}\ ;
\end{equation*}
this clearly corresponds to the arc-length distance on $S^1$ assuming that $S^1$ has total length $1$. Functions defined on $S^{1}$ can therefore be seen as one variable functions which are periodic of period $1$. Every algebraic and measure-related operation is therefore intended to be implicitly inherited from the structure over $\RR$. 
We consider our optimal control problem in the Banach spaces $E = L^{2}(S^{1};\RR)$ equipped with its usual norm $\norm{\cdot}_{2}$, and $E = C(S^{1};\RR)$ that is the space of real-valued continuous functions on $S^{1}$ equipped with the sup-norm  $\norm{\cdot}_{\infty}$. These two cases fall in the framework developed in \cite{boucekkine2019growth} and \cite{banachLattice2021} respectively, even though the latter is more general since it encompasses also the case $E = L^{2}(S^{1};\RR)$ by using the formalism of Banach lattices. 
We also introduce the notation $\ang{\cdot,\cdot}$ for the usual inner product of $L^{2}(S^{1};\RR)$, i.e. 
\begin{equation*}
\ang{f,g} := \int_{S^{1}}f(\theta)g(\theta) \,d\theta,\quad \forall f,g \in L^{2}(S^{1};\RR).
\end{equation*}
Consider therefore the \emph{positive cone} in $E$:
\begin{equation*}
E^{+} := \{x \in E\,|\, x(\cdot) \geq 0\},
\end{equation*}
and the \emph{strictly positive cone}: 
\begin{equation*}
E^{++} := \{x \in E\,|\, x(\cdot) > 0\}.
\end{equation*}
Moreover, we define $\RR^{+} = [0,+\infty)$ and $\RR^{++} = (0,+\infty)$ the sets of positive and strictly positive real numbers respectively.  

\subsection{The planner control problem} 
Consider the state variable is the level of capital $K(t,\theta)$ at time $t \in \RR^{+}$ and location $\theta \in S^{1}$. The evolution of capital over time is governed by the following partial differential equation:
\begin{equation}\label{eq:stateEquation}
\begin{cases}
\frac{\partial K}{\partial t}(t,\theta) = \sigma \frac{\partial^{2} K}{\partial \theta^{2}}(t,\theta) + A(\theta)K(t,\theta) - \eta(\theta) C(t,\theta),\\
K(0,\theta) = K_{0}(\theta),
\end{cases}
\end{equation}
with initial condition $K_{0}\in E^{+}$.
The constant $\sigma > 0$ measures the speed of reallocation of capital over space. 
Functions $A: S^{1}\to \RR^{++}$ and $\eta: S^{1}\to \RR^{++}$ represent the exogenous technological level and the exogenous population density respectively, which are assumed to be constant over time, space continuous, and heterogeneous between different locations $\theta \in S^{1}$.\footnote{Functions $A(\cdot)$ and $\eta(\cdot)$ are assumed to belong to $C(S^{1};\RR^{++})$. However, in the case $E = L^{2}(S^{1};\RR^{++})$ one does not need to require continuity, but only boundedness, i.e. $A,\eta \in L^{\infty}(S^{1};\RR^{++})$. }

Introduce the linear operator $L : D(L) \subseteq E \to E$ defined as
\begin{equation*}
(Lf)(\theta) = \sigma \frac{d^{2} f}{d \theta^{2}}(\theta) + A(\theta)f(\theta)
\end{equation*}
with $D(L) = C^{2}(S^{1};\RR)$.  One can check that, in the case $E = L^{2}(S^{1};\RR)$ or $E = C(S^{1};\RR)$ the operator $L$ is closed, densely defined and generates a $C_{0}$-semigroup $e^{tL}$ for $t \geq 0$ (see \cite{arendt1986one}). 

This allows us to define \emph{mild solutions} to Equation \eqref{eq:stateEquation}. In particular for any $K_{0}\in E$ and for any $C \in L^{1}_{loc}(R^{+};E^{+})$ we say that $K(t,\theta)$ is a mild solution to Equation \eqref{eq:stateEquation} if we can write
\begin{equation*}
K(t,\theta) = (e^{tL}K_{0})(\theta) - \int_{0}^{t}\left(e^{(t-s)L}\eta(\cdot)\,C(s,\cdot)\right)(\theta)\,ds \quad \text{ for all }t \geq 0.
\end{equation*}

The function $C(t,\theta)$ is the endogenous capital per capita consumption and has to be determined optimally. The objective functional the social planner aims to maximize is 
\begin{equation*}
J(C) :=  \int_{0}^{+\infty} e^{-\rho t}\left( \int_{S^{1}} \frac{C(t,\theta)^{1-\gamma}}{1-\gamma}\eta(\theta)^{q}\,d\theta \right)\,dt,
\end{equation*}
where $\rho > 0$ is the discount factor and the parameters $\gamma \in \RR^{++}\setminus\{1\}$ and $q \geq 0$ measure the intensity of risk aversion of the social planner.
The corresponding Value Function is:
\begin{equation*}
V_{+}(K_{0}) := \sup_{C \in \mathcal{A}_{+}(K_{0})} J(C),
\end{equation*} 
where the maximization of the objective ranges over the set of positive consumption strategies that maintain a positive level of capital in every location, i.e. we define the set
\begin{equation*}
\mathcal{A}_{+}(K_{0}) := \left\{ C \in L^{1}_{loc}(\RR^{+}; E^{+})\,|\, K(t,\theta) \geq 0 \forall (t,\theta) \in \RR^{+}\times S^{1} \right\}.
\end{equation*}
In what follows we are interested in the solution to the optimal control problem:
\begin{equation}\tag{P}\label{eq:problemP}
\text{maximize } J(C) \,\text{over the set} \,\mathcal{A}_{+}(K_{0})\, \text{under state equation}\, \eqref{eq:stateEquation}.
\end{equation}

\subsection{The auxiliary control problem}
In this section, we introduce the \emph{auxiliary control problem} associated with problem \eqref{eq:problemP}. This problem is characterized by a relaxed state constraint, namely we do not require the positivity of capital $K(t,\theta)$ at every period $t$.

Denote by $L^{*} : D(L^{*})\subseteq E^{*}\to E^{*}$  the adjoint operator of $L$. In particular, following \cite[Theorem 4.4]{banachLattice2021}, one can show that $L^{*}$ admits a simple eigenvalue $\lambda_{0}^{*}$ and an associated eigenvector $b_{0}^{*}$, such that 
\begin{equation*}
b_{0}^{*} \in E^{*}_{++} := \{\phi^{*} \in E^{*} \,|\, \ang{\phi^{*},f}  > 0 \,\forall f \in E^{+},f\not\equiv 0 \},
\end{equation*}
where $\ang{\cdot,\cdot}$ denotes the usual duality product. In the case $E = L^{2}(S^{1};\RR)$ we can easily identify by duality $b_{0}^{*}$ with a positive function in $L^{2}(S^{1};\RR)$. 

In the case $E = C(S^{1};\RR)$ the space $E^{*}$ is the space of Radon measures with bounded variation. The eigenvector $b_{0}^{*}$, therefore, is a strictly positive measure over the Borel sets of $S^{1}$, which is absolutely continuous with respect to the Lebesgue measure (see \cite[Proposition 5.6 and following sentences]{banachLattice2021}). Therefore, with a little abuse of notation, we will still denote by $b_{0}^{*}$ the corresponding density function. However, it is important for the well-posedness of our problem that $b_{0}^{*}$ is represented by a function belonging to $E$ itself, i.e. by a continuous function. This, however, is ensured by \cite[Proposition 5.6]{banachLattice2021} which shows that $b_{0}^{*}$ can be identified (up to a multiplicative constant) with the first eigenvector of the operator $L$, $b_{0}\in C(S^{1};\RR^{++})$ with corresponding eigenvalue $\lambda_{0}>0$, thus ensuring continuity.

We can now formulate the auxiliary problem. Introduce
\begin{equation}
J(C) := \int_{0}^{+\infty} e^{-\rho t}\left( \int_{S^{1}} \frac{C(t,\theta)^{1-\gamma}}{1-\gamma}\eta(\theta)^{q}\,d\theta \right)\,dt,
\end{equation}
where the feasible set is defined as:
\begin{equation}
\mathcal{A}^{b_{0}^{*}}_{+}(K_{0}) := \left\{ C \in L^{1}_{loc}(\RR^{+}; E^{+})\,|\, \ang{K(t,\cdot),b_{0}^{*}} \geq 0 \forall t \in \RR^{+} \right\},
\end{equation}
and the corresponding Value Function is: 
\begin{equation*}
V^{b_{0}^{*}}(K_{0}) := \sup_{C \in \mathcal{A}^{b_{0}^{*}}_{+}(K_{0})}J(C).
\end{equation*} 
We, therefore, define the auxiliary control problem as: 
\begin{equation}\tag{P aux}\label{eq:problemPauxiliary}
\text{maximize } J(C)\,\text{over the set}\,\mathcal{A}^{b_{0}^{*}}_{+}(K_{0})\, \text{under state equation}\, \eqref{eq:stateEquation}.
\end{equation}
\begin{rem}
It holds $\mathcal{A}_{+}(K_{0}) \subseteq \mathcal{A}^{b_{0}^{*}}_{+}(K_{0})$ since the constraint required by problem \eqref{eq:problemP} is more strict. Therefore, the corresponding value functions satisfy $V^{b_{0}^{*}}(K_{0}) \geq V_{+}(K_{0})$. 
\end{rem}
\begin{rem}
If the control $C^{*}$ is optimal for Problem \eqref{eq:problemPauxiliary} and it also satisfies the constraint of Problem \eqref{eq:problemP}, i.e. $C^{*}  \in A_{+}(K_{0})$, then it is optimal also for \eqref{eq:problemP}, and it holds $V^{b_{0}^{*}}(K_{0}) = V_{+}(K_{0})$.
\end{rem}

As previously stated, under the assumption that $\rho > \lambda_{0}(1-\gamma)$, problem \eqref{eq:problemPauxiliary} admits an explicit form for the optimal control, given by (see \cite[Theorem 5.7]{banachLattice2021})
\begin{equation}\label{eq:optimalConsumptionAux}
C^{*}(t,\theta) := \eta(\theta)^{\frac{q-1}{\gamma}} \left(\alpha b_{0}(\theta)\right)^{-\frac{1}{\gamma}}  \int_{S^{1}}K(t,\xi)b_{0}(\xi)\,d\xi,
\end{equation}
where 
\begin{equation*}
\alpha := \left(\frac{\gamma}{\rho - \lambda_{0}(1-\gamma)}\int_{S^{1}}\eta(\theta)^{\frac{q+\gamma-1}{\gamma}}b_{0}(\theta)^{\frac{\gamma-1}{\gamma}}\,d\theta \right)^{\gamma}.
\end{equation*}
The optimal capital distribution $K(t,\theta)$ for problem \eqref{eq:problemPauxiliary}  therefore solves
\begin{equation}\label{eq:optimalProblem}
\begin{cases}
\frac{\partial K}{\partial t}(t,\theta) = \sigma \frac{\partial^{2} K}{\partial \theta^{2}}(t,\theta) + A(\theta)K(t,\theta) - \eta(\theta)^{\frac{q+\gamma-1}{\gamma}}\left(\alpha b_{0}(\theta)\right)^{-\frac{1}{\gamma}}  \int_{S^{1}}K(t,\xi)b_{0}(\xi)\,d\xi,\\
K(0,\theta) = K_{0}(\theta).
\end{cases}
\end{equation}
\begin{rem}
Equation \eqref{eq:optimalProblem} can also be rewritten by using the linear operator $F:D(F) \subseteq E \to E$ defined for $\theta \in S^{1}$ as 
\begin{equation}\label{eq:operatorF}
(Fx)(\theta) := \sigma \frac{\partial^{2}x}{\partial \theta^{2}}(\theta) + A(\theta)x(\theta) 
-\eta(\theta)^{\frac{q+\gamma-1}{\gamma}} \left(\alpha b_{0}(\theta)\right)^{\frac{1}{\gamma}}\left( \int_{S^{1}} x(\xi)b_{0}(\xi)\, d\xi \right)\eta(\theta)^{\frac{q+\gamma-1}{\gamma}}, 
\end{equation}
as a differential equation in Banach spaces
\begin{equation}\label{eq:optimalProblemF}
\begin{cases}
K'(t) = FK(t),\\
K(0) = K_{0}.
\end{cases}
\end{equation}
\end{rem}
\begin{rem}
The optimal control policy can also be expressed as a function only of the initial condition $K_{0}(\theta)$ as 
\begin{equation*}
C^{*}(t,\theta) := \eta(\theta)^{\frac{q-1}{\gamma}}\left(\alpha b_{0}(\theta)\right)^{-\frac{1}{\gamma}}  \int_{S^{1}}K_{0}(\xi)\xi)b_{0}(\xi)\,d\xi \cdot e^{gt}
\end{equation*}
where $g := (\lambda_{0}-\rho)/\gamma$ is the optimal aggregate growth rate.
\end{rem}
\begin{rem}
The optimal consumption given by \eqref{eq:optimalConsumptionAux} has the notable property of being nonlocal and space independent with respect to $K(t,\theta)$. The only space heterogeneity in the optimal consumption strategy is due to the heterogeneity of $\eta(\theta)$ and $b_{0}(\theta)$ which however are constant over time and independent of $K(t,\theta)$. This space independence on the level of capital is crucial to show that in certain regimes the optimal capital distribution may become negative, thus showing the non-optimality of the consumption strategy in certain scenarios. 
\end{rem}

Referring to \cite[Theorem 5.12]{banachLattice2021}, one can prove that when the initial condition $K_{0}$ of \eqref{eq:optimalProblem} is close enough to the steady state solution of the detrended problem, i.e to the limit for large $t$ of $K_{g}(t,\theta):= e^{-gt}K(t,\theta)$, then the solution will remain positive for all $t > 0$. This in particular implies that the control $C^{*}$ is also optimal for Problem \eqref{eq:problemP} since $C^{*}  \in A_{+}(K_{0})$.
The possibility of proving such results has been made possible by developing the theory in $C(S^{1};\RR)$, since the usual $L^{p}$-theory doesn't allow for pointwise estimations. 

However the general case, without any assumptions on the initial condition, remains open. It comes therefore natural to investigate whether or not the solution of equation \eqref{eq:optimalProblem} will remain positive both in the case $E = L^{2}(S^{1};\RR)$ and $E = C(S^{1};\RR)$. This motivates the study of the invariance properties of Equation \eqref{eq:optimalProblem} that we carry out in Section \ref{sec:noninvariance}.

\section{Non-invariance in the positive cone}\label{sec:noninvariance}
In this section, we show that Equation \eqref{eq:optimalProblem} is not invariant over the positive cone $E^{+}$, which is a convex closed set,  in $E = L^{2}(S^{1};\RR)$ and $E = C(S^{1};\RR)$. We will adopt the framework developed in \cite{cannarsa2018invariance} which provides sufficient and necessary conditions for the invariance of the flow of a dynamical system in a convex and closed set in a separable Banach space. Therefore we will reformulate our original problem in the setting of differential equations in abstract vector spaces. 

Consider the Cauchy problem in the separable Banach space $(E,\norm{\cdot})$:
\begin{equation}\label{eq:abstractODEcannarsa}
\begin{cases}
X'(t) = \mathscr{A}X(t) + \mathscr{B}(X(t))\,, t \geq 0,\\
X(0) = x.
\end{cases}
\end{equation} 
Hypothesis $(\textbf{H})$: assume moreover that $\mathscr{A} : D(\mathscr{A}) \subseteq E \to E$ is the infinitesimal generator of a strongly continuous semigroup of contractions, and that $\mathscr{B}:E\to E$ is continuous and quasi-dissipative.

To reformulate our problem, defined in Equation \eqref{eq:optimalProblem}, in this setting we therefore consider (recall the definition of $F$ in Equation \eqref{eq:operatorF}):
\begin{equation*} 
    \mathscr{A} := F - \alpha I  \, \text{ and }\,  \quad  \mathscr{B} := \alpha I,
\end{equation*}
with $\alpha > \max_{\theta \in S^{1}} \abs{A(\theta)}$, so that $\mathscr{A}$ is a dissipative operator, and $\mathscr{B}$ is quasi dissipative so the aforementioned hypothesis on $\mathscr{A}$ and $\mathscr{B}$ are fulfilled (other parts of the hypothesis are readily verified).

\subsection{General strategy}
The key ingredients provided in \cite{cannarsa2018invariance} are the \emph{distance from a closed convex set} $G \subseteq E$:
\begin{equation*}
d_{G}(x) := \inf_{y \in G} \norm{x-y},
\end{equation*}
and the \emph{lower Dini derivative} of a locally Lipschit function $f$ in a point $x \in E$ along direction $v$:
\begin{equation*}
D^{-}f(x)v := \lim_{\lambda \to 0}\frac{f(x+\lambda v)-f(x)}{\lambda}, \quad v \in E.
\end{equation*}
The idea behind how these two elements combine is the following: if one shows that the distance to the candidate invariant set $d_{G}(x)$ along the direction of the dynamic of Equation \eqref{eq:abstractODEcannarsa} $\mathscr{A}x + \mathscr{B}(x)$ decreases exponentially, then one can apply a Grownall type argument to show that the boundary is never reached. This intuition is made feasible to prove the desired result by considering the distance from the complement of the candidate invariant set; introduce 
\begin{equation*}
G_{\delta} := \{x \in E \,|\, d_{G}(x) < \delta\}
\end{equation*}
which is an open set contained in the complement of $G$. 

We refer to \cite[Theorem 4.2]{cannarsa2018invariance}, that reads:
\begin{thm*}[{\cite[Theorem 4.2]{cannarsa2018invariance}}]
Assume $(\mathbf{H})$. Then $G$ is invariant if and only if there exist $\delta > 0$ and $C > 0$ such that
\begin{equation*}
D^{-}d_{G}(x)(\mathscr{A}x+\mathscr{B}(x)) \leq C d_{G}(x) \quad  \forall x \in G_{\delta}.
\end{equation*}
\end{thm*}
In the next sections, we specialize the analysis to the case $E = L^{2}(S^{1};\RR)$ and $E = C(S^{1};\RR)$.

\subsection{The case $E = L^{2}(S^{1};\RR)$}
In this section we will denote by $\ang{\cdot,\cdot}$ the scalar product in $L^{2}(S^{1};\RR)$. 
In the special case where $E$ is a separable Hilbert space, the aforementioned invariance results can be specialized, by rewriting the distance with respect to a set using the orthogonal projection and by making explicit the lower Dini derivative. In particular, we have that the distance from the set $G$ is given by $d_{G}(x) = \norm{x-\Pi_{G}(x)}_{2}$, where $\Pi_{G}(x)$ is the orthogonal projection over $G$. Moreover one can express the lower Dini derivative as
\begin{equation*}
D^{-}d_{G}(x)v = \ang{\nabla d_{G}(x),v},\quad \text{where }\, \nabla d_{G}(x) = \frac{x-\Pi_{G}(x)}{d_{G}(x)}, x \in E\setminus G.
\end{equation*}
In the particular case where $G = E^{+}$ one also has 
\begin{equation*}
\Pi_{G}(x) = x^{+},\quad x-\Pi_{G}(x) = -x^{-},\quad d_{G}(x) = \norm{x^{-}}_{2},
\end{equation*}
by denoting with $x^{+}$ and $x^{-}$ the standard positive and negative part. 
It is, therefore, possible to reformulate \cite[Theorem 4.2]{cannarsa2018invariance} in a simpler way, see \cite[Corollary 4.4]{cannarsa2018invariance}. Since we are interested in showing that $E^{+}$ is not invariant for the flow of Equation \eqref{eq:optimalProblemF}, we are interested in the converse result, which reads:
\begin{prop}[{Converse of \cite[Corollary 4.4]{cannarsa2018invariance}}]
\label{prop:ConverseCor44Cannarsa}
Assume hypothesis $(\mathbf{H})$. Then the positive cone $G = E^{+}$ is not invariant if and only if for all $\delta,C > 0$ there exist $x \in G_{\delta}$ such that 
\begin{equation*}
-\ang{x^{-},\mathscr{A}x + \mathscr{B}(x)} > C \norm{x^{-}}_{2}^{2}.
\end{equation*}
\end{prop}

\begin{thm}\label{teo:noninvarianceL2}
The positive cone $G = E^{+}$ is not invariant for equation \eqref{eq:optimalProblem} in $E = L^{2}(S^{1};\RR)$.
\end{thm}
\begin{proof}
We use Proposition \ref{prop:ConverseCor44Cannarsa}, that is we have to show that 
\begin{equation*}
\forall C,\delta > 0\, \exists f  \in G_{\delta}\cap D(\mathscr{A}) \text{ (depending on $C$ and $\delta$) such that } 
\end{equation*}
\begin{equation*}
-\ang{f^{-},\mathscr{A}f + \mathscr{B}(f)} > C \norm{f^{-}}_{2}^{2}.
\end{equation*}
In our case, we need to prove that such an $f$ satisfies 
\begin{multline*}
-\int_{S^{1}}f^{-}(\theta)\left[\sigma \frac{\partial^{2}f}{\partial \theta^{2}}(\theta) + A(\theta)f(\theta) 
- \eta(\theta)^{\frac{q+\gamma-1}{\gamma}}\left(\alpha b_{0}(\theta)\right)^{-\frac{1}{\gamma}}\int_{S^{1}} f(\xi)b_{0}(\xi)\, d\xi \right] \, d\theta  \\ > C \int_{S^{1}}\abs{f^{-}(\theta)}^2\,d \theta.
\end{multline*}
We can study each term on the left-hand side of the previous inequality separately: the first one can be simplified as follows:
\begin{multline*}
-\sigma\int_{S^{1}}f^{-}(\theta) \frac{\partial^{2}f}{\partial \theta^{2}}(\theta)\,d\theta = \\
-\sigma\int_{S^{1}}f^{-}(\theta) \frac{\partial^{2}}{\partial \theta^{2}}[f^{+}(\theta)-f^{-}(\theta)]\,d\theta = \\
\sigma \int_{S^{1}}\frac{\partial}{\partial \theta}f^{-}(\theta)\frac{\partial}{\partial \theta}f^{+}(\theta)\,d\theta - \sigma \int_{S^{1}}\abs{\frac{\partial}{\partial \theta}f^{-}(\theta)}^{2}\,d\theta=\\
- \sigma \int_{S^{1}}\abs{\frac{\partial}{\partial \theta}f^{-}(\theta)}^{2}\,d\theta,
\end{multline*}
since the supports of $f^{+}$ and $f^{-}$ are disjoint.
The second term is rewritten as: 
\begin{multline*}
-\int_{S^{1}}A(\theta)f^{-}(\theta)f(\theta)\,d\theta = 
-\int_{S^{1}}A(\theta)f^{-}(\theta)[f^{+}(\theta)-f^{-}(\theta)]\,d\theta = \\ = 
\int_{S^{1}}A(\theta)\abs{f^{-}(\theta)}^{2}\,d\theta \geq \underline{A}\int_{S^{1}}\abs{f^{-}(\theta)}^{2}\,d\theta,
\end{multline*}
where $\underline{A}:= \inf_{\theta \in S^{1}}A(\theta)$.
Finally, for the third term, introduce the notation:
\begin{equation}\label{eq:defPsi}
\psi(\theta):= \eta(\theta)^{\frac{q+\gamma-1}{\gamma}}\left(\alpha b_{0}(\theta)\right)^{-\frac{1}{\gamma}},
\end{equation}
which is a strictly positive bounded continuous function, by the assumptions on $\eta(\cdot)$ and the results on $b_{0}(\cdot)$.
Therefore we have 
\begin{equation*}
\psi(\theta) \int_{S^{1}}f^{-}(\theta)\int_{S^{1}}f(\xi)b_{0}(\xi)\,d\xi\,d\theta \geq
\underline{\psi}\, \underline{b_{0}}\int_{S^{1}}f^{-}(\theta)\,d\theta\int_{S^{1}}f(\theta)\,d\theta,
\end{equation*}
where $\underline{\psi}:= \inf_{\theta \in S^{1}}\psi(\theta)$, and $\underline{b_{0}}:= \inf_{\theta \in S^{1}}b_{0}(\theta)$.

Summarizing, we have to show that $\forall C,\delta > 0$ there exists an $f \in G_{\delta} \cap D(A)$ such that
\begin{multline*}
-\sigma  \int_{S^{1}}\abs{\frac{\partial}{\partial \theta} f^{-}(\theta)}^2\,d\theta + \underline{A}\int_{S^{1}}\abs{ f^{-}(\theta)}^2\,d\theta + \underline{\psi}\, \underline{b_{0}} \int_{S^{1}}f^{-}(\theta)\,d\theta\int_{S^{1}}f(\theta)\,d\theta > \\  C \int_{S^{1}}\abs{ f^{-}(\theta)}^2\,d\theta,
\end{multline*}
or equivalently
\begin{equation*}
(C-\underline{A})\int_{S^{1}}\abs{f^{-}(\theta)}^2\,d\theta + \sigma \int_{S^{1}}\abs{\frac{\partial}{\partial \theta} f^{-}(\theta)}^2\,d\theta < \underline{\psi}\, \underline{b_{0}} \int_{S^{1}}f^{-}(\theta)\,d\theta\int_{S^{1}}f(\theta)\,d\theta.
\end{equation*}
In the previous one can understand why given $C$ and $\delta$ one can construct a function $f$ that satisfies the inequality: on the left-hand side, there are only terms related to the function $f^{-}$, while on the right-hand side, it appears the expression $\int f$. This suggests that if one can make  $\int f$ arbitrarily large while leaving unchanged $f^{-}$ then the inequality can be satisfied. 

\begin{figure}
\centering
\includegraphics[width=0.75\textwidth]{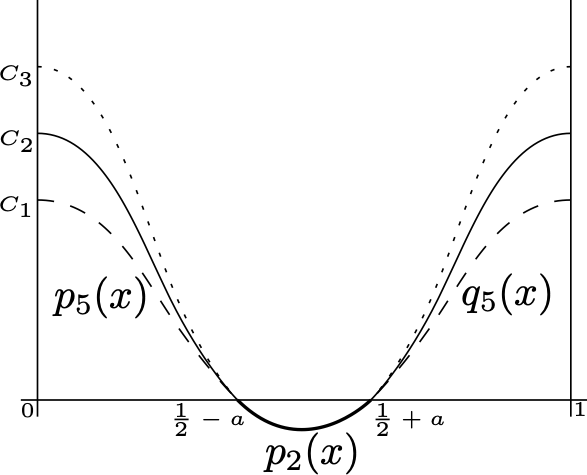}
\caption{Figure of the possible functions $f$ used in the proof for different values of the constant $C_{*}$.}
\label{fig:picTheoCounterExample}
\end{figure}

One possible example can be obtained in the following way, referring to Figure \ref{fig:picTheoCounterExample}. For any $C$ and $\delta$ we define $f$ piecewise in the intervals $[0,1/2-a], (1/2-a,1/2+a)$ and $[1/2+a,1)$, where the constant $a$ will be chosen later, as 
\begin{equation*}
f(x) = \begin{cases}
p_{5}(x) \text{ for } x \in [0,1/2-a],\\
p_{2}(x) \text{ for } x \in [1/2-a,1/2+a],\\
q_{5}(x) \text{ for } x \in [1/2+a,1].
\end{cases}    
\end{equation*}
Now take $a = \left(\frac{\sqrt{15}}{8}\delta\right)^{\frac{2}{5}}$, and define the middle piece $p_{2}(x)$ as
\begin{equation*}
p_{2}(x) = (x-\frac{1}{2}+a)(x-\frac{1}{2}-a).
\end{equation*}
Notice that by the choice of $a$ on has $\norm{p_{2}}_{L^{2}((1/2-a,1/2+a))} = \delta/2<\delta$.

The functions $p_{5}$ and $q_{5}$ are instead defined implicitly as polynomials of degree five, satisfying the following conditions:
\begin{equation*}
\begin{cases}
p_{5}(\frac{1}{2}-a) = 0\\
p^{'}_{5}(\frac{1}{2}-a) = p^{'}_{2}(\frac{1}{2}-a)\\
p^{''}_{5}(\frac{1}{2}-a) = p^{''}_{2}(\frac{1}{2}-a)\\
p_{5}(0) = C_{*}\\
p^{'}_{5}(0) = 0\\
p^{''}_{5}(0) = 0\\
\end{cases},\,
\begin{cases}
q_{5}(\frac{1}{2}+a) = 0\\
q^{'}_{5}(\frac{1}{2}+a) = p^{'}_{2}(\frac{1}{2}+a)\\
q^{''}_{5}(\frac{1}{2}+a) = p^{''}_{2}(\frac{1}{2}+a)\\
q_{5}(1) = C_{*}\\
q^{'}_{5}(1) = 0\\
q^{''}_{5}(1) = 0\\
\end{cases},
\end{equation*}
where the choice of $C_{*} > 0$ is arbitrary. Such functions $p_{5}$ and $q_{5}$ must exist since they are defined by six independent conditions on the coefficients. 
The resulting function $f$ so constructed belongs to $C^{2}(S^{1};\RR)$, and satisfies 
\begin{equation*}
f^{-}(x) = \begin{cases}
-p_{2}(x)\quad &\text{ for } x \in [1/2-a,1/2+a],\\
0 \quad &\text{ otherwise, }
\end{cases}
\end{equation*}
so that it holds $\norm{f^{-}}_{2} = \delta/2$. 

Now we see that depending on the value of $C_{*}$ the quantity $\int f$ can be made arbitrarily large, while leaving unchanged $f^{-}$ and consequently also $\frac{\partial}{\partial \theta} f^{-}$. This ends the proof
\end{proof}

\subsection{The case $E = C(S^{1};\RR)$}
In this setting, we are not able to apply the easier formulation available in Hilbert spaces. Therefore, we have to resort to \cite[Theorem 4.2]{cannarsa2018invariance}, or more precisely to the converse of that result, that in our case reads: 
\begin{thm*}[Converse of {\cite[Theorem 4.2]{cannarsa2018invariance}}]
Assume $(\mathbf{H})$. Then $G$ is not invariant if and only if for any $\delta,C > 0$ there exists $x \in G_{\delta}$ such that
\begin{equation*}
D^{-}d_{G}(x)(\mathscr{A}x+\mathscr{B}(x)) \geq C d_{G}(x).
\end{equation*} 
\end{thm*}
We recall that in the case where $G = E^{+}$ the positive cone, one has
\begin{equation*}
d_{G}(x) = \norm{x^{-}}_{\infty}.
\end{equation*}

\begin{thm}\label{teo:noninvarianceC}
The positive cone $G = E^{+}$ is not invariant for equation \eqref{eq:optimalProblem} in $E = C(S^{1};\RR)$.
\end{thm}
\begin{proof}
The strategy of the proof is similar to that of Theorem \ref{teo:noninvarianceL2}.

We start by applying the computation in \cite[Example 5.3]{cannarsa2018invariance} with $f \equiv 0$. Consider $(h_{n})_{n \in \NN}$  a positive sequence going to zero: we can compute the lower Dini derivative of the distance from the set $G$, in any direction $v \in E$, as 
\begin{equation*}
D^{-}d_{G}(x)v := \lim_{n \to \infty} \frac{d_{G}(x+h_{n}v)-d_{G}(x)}{h_{n}} = \frac{\norm{(x+h_{n}v)^{-}}_{\infty}-\norm{x^{-}}_{\infty}}{h_{n}}. 
\end{equation*}
For each $n$ consider $\theta_{n}\in S^{1}$ an element which realizes the maximum 
\begin{equation*}
\theta_{n} \in \arg\max_{\theta \in S^{1}}(x+h_{n}v)^{-}(\theta).
\end{equation*}
Moreover, for $n$ sufficiently large, it holds
\begin{equation*}
\norm{(x+h_{n}v)^{-}}_{\infty} = x^{-}(\theta_{n}) - h_{n}v(\theta_{n}).
\end{equation*}
Notice in fact that for $n$ large the sign of $x+h_{n}v$ is the same as the sign of $x$ since $h_{n} \to 0$. Moreover, since $\norm{x^{-}} \neq 0$ it also holds $x(\theta_{N}) < 0$ definetively. Therefore 
\begin{equation*}
\norm{(x+h_{n}v)^{-}}_{\infty} = -x(\theta_{n}) - h_{n}v(\theta_{n}) = x^{-}(\theta_{n}) - h_{n}v(\theta_{n}).
\end{equation*} 
Hence we have 
\begin{equation*}
\norm{(x+h_{n}v)^{-}}_{\infty} \leq \norm{x^{-}} - h_{n}v(\theta_{n}),
\end{equation*}
which implies for $n$ large enough
\begin{equation*}
D^{-}d_{G}(x)v \leq -v(\theta_{n}).
\end{equation*}
By compactness, we can assume that $\theta_{n} \to \theta_{0} \in S^{1}$ as $n$ goes to infinity, and by the continuity of the function $x$ we conclude that 
\begin{equation*}
\theta_{0} \in \arg\max_{\theta \in S^{1}}x^{-}(\theta),
\end{equation*}
and therefore
\begin{equation*}\label{eq:inequalityDistanceProofC}
D^{-}d_{G}(x)v \leq -v(\theta_{0}) \leq \max\{-v(\theta)\,|\, \theta \in \arg\max_{\theta \in S^{1}}x^{-}(\theta)\}.
\end{equation*}
Consequently, it also holds, by choosing $z = -v$
\begin{multline*}
-D^{-}d_{G}(x)v = D^{-}d_{G}(x)(-v) = D^{-}d_{G}(x)z \leq \max\left\{ -z(\theta) \,\big\vert \, \theta \in \arg\max_{\theta \in S^{1}} x^{-}(\theta)\right\} \phantom{=} \\ = \max\left\{ v(\theta) \,\big\vert \, \theta \in \arg\max_{\theta \in S^{1}} x^{-}(\theta)\right\},
\end{multline*}
from which it follows
\begin{equation*}
-D^{-}d_{G}(x)v \leq \max\left\{ v(\theta) \,\big\vert \, \theta \in \arg\max_{\theta \in S^{1}} x^{-}(\theta)\right\},
\end{equation*}
or equivalently
\begin{equation*}
D^{-}d_{G}(x)v \geq -\max\left\{ v(\theta) \,\big\vert \, \theta \in \arg\max_{\theta \in S^{1}} x^{-}(\theta)\right\} = \min\left\{ -v(\theta) \,\big\vert \, \theta \in \arg\max_{\theta \in S^{1}} x^{-}(\theta)\right\}.
\end{equation*}
Therefore we recover the following sufficient condition for the non-invariance of the set $G$: for any $\delta, C>0$ there exists $f \in G_{\delta}$ such that
\begin{multline}\label{eq:minNonInvarianceC}
\min\left\{ -\sigma \frac{\partial^{2}f}{\partial \theta^{2}}(\theta) - A(\theta)f(\theta) + \psi(\theta)\int_{S^{1}} f(\xi)b_{0}(\xi)\, d\xi  \,\big\vert\, \theta \in \arg\max_{\theta \in S^{1}} f^{-}(\theta)\right\}\geq  \\ \geq C\norm{f^{-}}_{\infty},
\end{multline}
by using the definition of $\psi(\theta)$ introduced in Equation \eqref{eq:defPsi}.

As in the proof of Theorem \ref{teo:noninvarianceL2}, from the previous one can understand why given $\delta$ and $C$ it is possible to find such an $f$: on the left-hand side we see that the nonlocal term $ \int_{S^{1}}f(\xi)\,d\xi$ does not depend on $\theta \in \arg\max_{\theta \in S^{1}} f^{-}(\theta)$, so it is possible to make it arbitrarily large, without changing $f^{-}$. 

Let us consider the following function $f = f^{\delta, C}$, that is analogous to the one used in the proof of Theorem \ref{teo:noninvarianceL2} (refer to Figure \ref{fig:picTheoCounterExample}): 
\begin{equation*}
f(x) = \begin{cases}
p_{5}(x) \text{ for } x \in [0,1/2-a],\\
p_{2}(x) \text{ for } x \in [1/2-a,1/2+a],\\
q_{5}(x) \text{ for } x \in [1/2+a,1].
\end{cases}    
\end{equation*}
Now take $a = \sqrt{\frac{\delta}{2}}$, and define the middle piece $p_{2}(x)$ as
\begin{equation*}
p_{2}(x) = (x-\frac{1}{2}+a)(x-\frac{1}{2}-a).
\end{equation*}
Notice that by the choice of $a$ one has $\norm{p_{2_{|{(1/2-a,1/2+a)}}}}_{\infty} = \delta/2<\delta$. The functions $p_{5}$ and $q_{5}$ are instead defined implicitly as polynomials of degree five, satisfying the following conditions:
\begin{equation*}
\begin{cases}
p_{5}(\frac{1}{2}-a) = 0\\
p^{'}_{5}(\frac{1}{2}-a) = p^{'}_{2}(\frac{1}{2}-a)\\
p^{''}_{5}(\frac{1}{2}-a) = p^{''}_{2}(\frac{1}{2}-a)\\
p_{5}(0) = C_{*}\\
p^{'}_{5}(0) = 0\\
p^{''}_{5}(0) = 0\\
\end{cases},\,
\begin{cases}
q_{5}(\frac{1}{2}+a) = 0\\
q^{'}_{5}(\frac{1}{2}+a) = p^{'}_{2}(\frac{1}{2}+a)\\
q^{''}_{5}(\frac{1}{2}+a) = p^{''}_{2}(\frac{1}{2}+a)\\
q_{5}(1) = C_{*}\\
q^{'}_{5}(1) = 0\\
q^{''}_{5}(1) = 0\\
\end{cases},
\end{equation*}
where the choice of $C_{*} > 0$ is arbitrary.

It is trivial to see that the only absolute maximum of $f^{-}$ is attained in $\xi_{0} = \frac{1}{2}$ so that the minimum on the left-hand side of \eqref{eq:minNonInvarianceC} can be omitted since the set is a singleton, and one has only to prove that 
\begin{equation*}
-\sigma \frac{\partial^{2}f}{\partial \theta^{2}}(\xi_{0}) - A(\xi_{0})f(\xi_{0}) + \psi(\xi_{0})\int_{S^{1}} f(\xi)b_{0}(\xi)\, d\xi \geq C\delta/2.
\end{equation*}
From the previous, we see that by changing the value of $C^{*}$ in the definition of $p_{5}$ and $q_{5}$, we can make $\psi(\xi_{0})\int_{S^{1}} f(\xi)b_{0}(\xi)\, d\xi $ arbitrarily large (recall $\psi(\cdot) > 0$), and leaving untouched all the other terms since $\xi_{0}$ and $f(\xi_{0})$ are unaffected by the choice of $C^{*}$. This ends the proof.
\end{proof}
\begin{rem}
Observe that in Theorem \ref{teo:noninvarianceL2} and Theorem \ref{teo:noninvarianceC} there are no additional assumptions on the parameters to prove the non-invariance results. This is coherent with the proof where it is shown that the only important element that matter is the total aggregate capital level. This intuition will be also confirmed in Section \ref{sec:numerical} via numerical simulations. 
\end{rem}

\section{Numerical experiments in the positive cone}\label{sec:numerical}
In this Section, we try to further characterize what we have proven in the previous section by providing some numerical examples. In particular, we look at the case where $E = C(S^{1};\RR)$ and the initial condition lies in $E^{+}$. For the sake of simplicity, we only consider examples where both the exogenous technological progress $A(\theta)$ and the population density $\eta(\theta)$ are assumed to be space homogeneous, i.e. $A(\theta) \equiv A > 0$ and $\eta(\theta)\equiv \eta > 0$.  In the two examples that we will analyze we will consider two different initial conditions $K_{0}$, which differ from each other only in terms of the total integral $\int_{S^{1}}K_{0}(\xi)\,d\xi$. 

Introduce the polynomial 
\begin{equation*}
p_{3}(x) = 1 - 6x^{2} + 8x^{3} - 3x^{4}
\end{equation*}
which satisfies $p_{3}(0) = p_{3}(1) = p_{3}'(0) = p_{3}'(1) = 0$. 
Consider therefore, for $R = 1/4$ and $\epsilon = 1/10$, the function 
\begin{equation*}
f_{0}(\theta) = 
\begin{cases}
1 \quad &\text{ for } \theta \in (\frac{1}{2}-R+\frac{\epsilon}{2},\frac{1}{2}+R-\frac{\epsilon}{2}),\\
0 \quad &\text{ for } \theta \in [0,\frac{1}{2}-R-\frac{\epsilon}{2})\cup (\frac{1}{2}+R+\frac{\epsilon}{2},1),\\
p_{3}(\frac{1}{\epsilon}(\theta-\frac{1}{2}-R+\frac{\epsilon}{2})) &\text{ for } \theta \in [\frac{1}{2}+R-\frac{\epsilon}{2},\frac{1}{2}+R+\frac{\epsilon}{2}], \\
p_{3}(\frac{1}{\epsilon}(-\theta+\frac{1}{2}-R+\frac{\epsilon}{2})) &\text{ for } \theta \in [\frac{1}{2}-R-\frac{\epsilon}{2},\frac{1}{2}-R+\frac{\epsilon}{2}]. \\
\end{cases}
\end{equation*}
Roughly speaking the function $f_{0}(\theta)$ is a continuous function over $S^{1}$ which is zero around the point $1/2$, takes the value $1$ in a neighborhood of the identified extremes $0,1\in S^{1}$ and is interpolated continuously elsewhere. 

We thus consider $K_{0}(\theta) := \overline{K}f_{0}(\theta)$ for different values of $\overline{K} > 0$. This situation corresponds to having an initial poorer region around $\theta = 1/2$ and a richer one closer to the extremes. The value of $\overline{K}$ modulates the level of inequality between the poor and the rich region of our domain. In particular, it also affects the total level of aggregate capital $\int_{S^{1}}K_{0}(\xi)\,d\xi$ that is present in the system at time $t = 0$. As previously observed in the proof of Theorems \ref{teo:noninvarianceL2} and \ref{teo:noninvarianceC}, the possibility of making large the (negative) term related to the optimal consumption, whilst leaving unchanged all the other terms of Equation \eqref{eq:optimalProblem} is the key to finding negative values in the capital level $K(t,\theta)$ for $t > 0$. 

All the simulations that we will perform will range only for a short period $T = 1$. Since we are starting with an initial condition that is zero for some $\theta \in S^{1}$, we expect that if negative values arise, those will appear close to the initial time $t = 0$. After a few periods, the diffusion effect will take place, equalizing quickly the level of capital along space.\footnote{This is true since we are working in the space homogeneous case $A(\theta) \equiv A > 0$ and $\eta(\theta)\equiv \eta > 0$.}  We refer to Figure \ref{fig:numericlSimulation} for the considered simulated scenarios. All the parameters used in the simulations are summarized in Table \ref{tab:parameterValues}.
\begin{table}[!htbp]
\centering
\begin{tabular}{rcl}
\hline \hline
Parameters & Value & Description \\	
\hline	
$T$ & $1$  & Time horizon of simulation\\
$A$ & $10^{-2}$  & Exogenous technological progress\\
$\eta$ & $10^{-2}$  & Exogenous population density\\
$\sigma$  & $10^{-2}$ & Diffusion coefficient\\
$q$ & 1 & Risk aversion - population density\\
$\gamma$ & $1/2$ & Risk aversion - consumption\\
$R$ & $1/4$ & Radius of the poor region\\
$\epsilon$ & $1/10$ & Smoothing coefficient poor region\\
$\overline{K}$ & $10,100$ & Scaling factor of initial capital level\\
\hline \hline
\end{tabular}
\caption{The baseline values of the model's parameter used for numerical explorations.}
\label{tab:parameterValues}
\end{table}

\begin{figure}[H]
\centering
\begin{subfigure}[b]{0.47\textwidth}
\centering
\includegraphics[width=\textwidth]{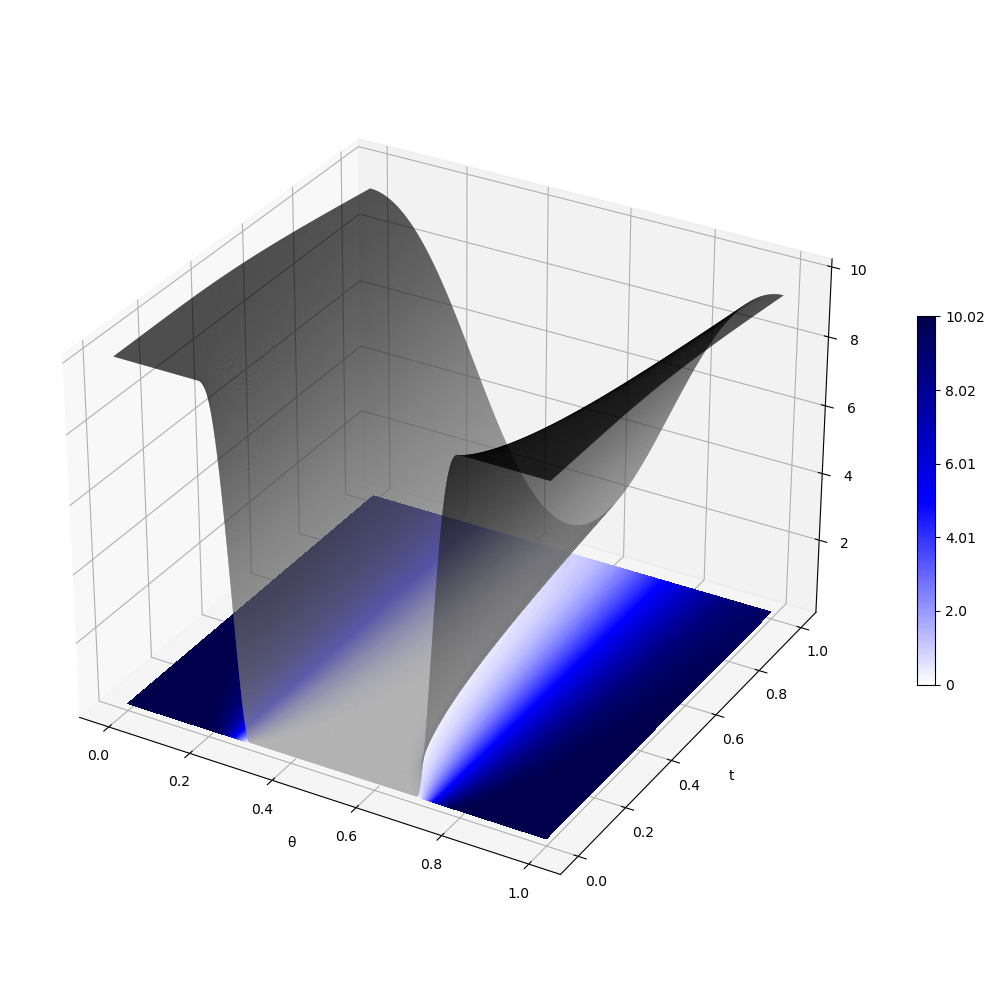}
\caption{$\overline{K} = 10$.}
\label{fig:numericlSimulation10}	
\end{subfigure}
\begin{subfigure}[b]{0.47\textwidth}
\centering
\includegraphics[width=\textwidth]{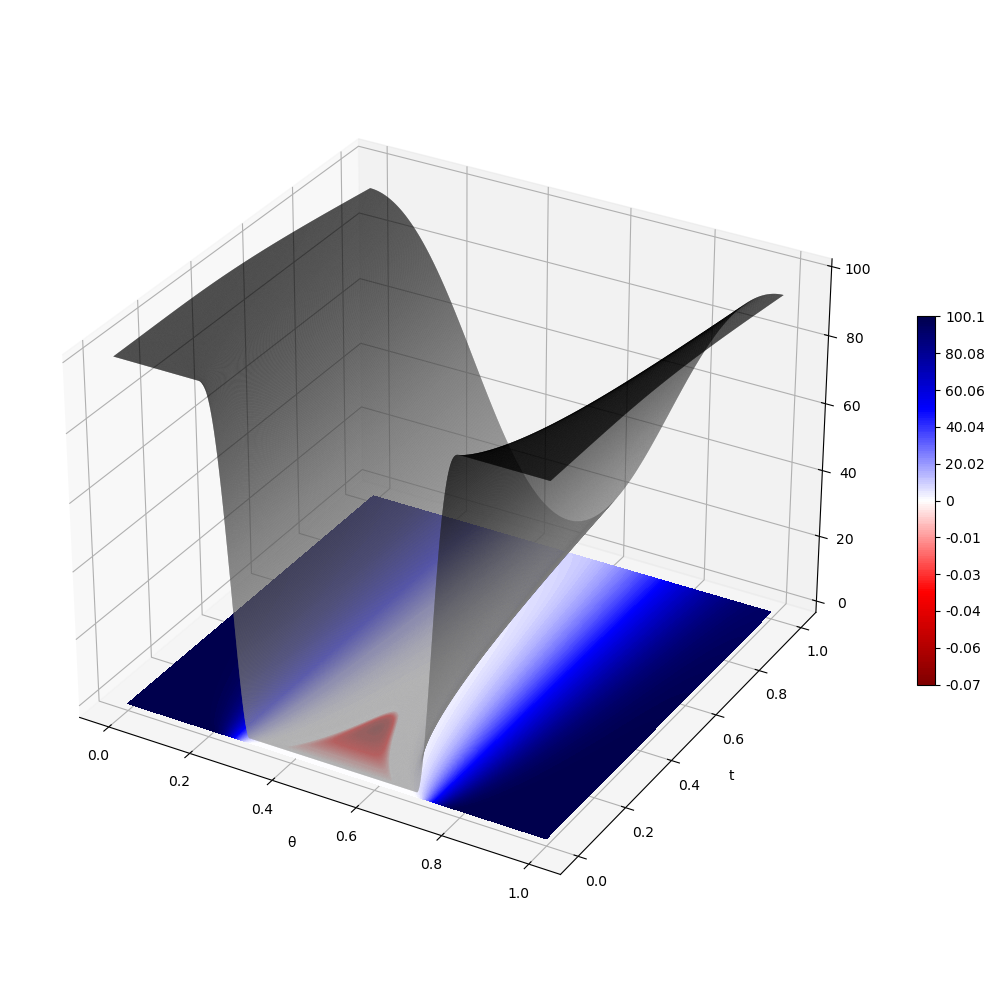}
\caption{$\overline{K} = 100$.}
\label{fig:numericlSimulation100}	
\end{subfigure}
\caption{Numerical results for different values of $\overline{K}$. The colormap below the surface plots shows positive (blue) and negative (red) values of the solution. } 	
\label{fig:numericlSimulation}	
\end{figure}

Figure \ref{fig:numericlSimulation10} shows the resulting capital allocation for $\overline{K} = 10$. We see in particular that, even though the initial capital profile is zero in some regions, the corresponding optimal capital allocation will never be negative. On the contrary in Figure \ref{fig:numericlSimulation100}, for $\overline{K} = 100$, we see that capital gets negative for a brief period and comes back positive everywhere afterward. This confirms the idea that, independently of the choice of the parameters of the model, Theorems \ref{teo:noninvarianceL2} and Theorem \ref{teo:noninvarianceC} hold. In particular, for any choice of the model parameters it always exists an initial capital allocation for which the capital distribution will become negative for some positive time.  However, a full characterization of the relation between the model parameter and the initial condition that leads to negative values of the solution is still lacking. 

\section{Conclusions}\label{sec:conclusion}
In this paper, we provided a rigorous proof of the non-invariance of the positive cone for the dynamic of equation \eqref{eq:optimalProblem} in the cases $E = L^{2}(S^{1};\RR)$ and $E = C(S^{1};\RR)$. This in particular shows that, if the initial capital allocation of Problem \eqref{eq:problemP} assumes the value zero in some part of the domain $S^{1}$, then the consumption strategy defined in Equation \eqref{eq:optimalConsumptionAux} is not optimal, since it does not satisfy the state constraint $K(t,\theta) \geq 0$.
This in a sense provides an opposite (negative) result, with respect to the one proved in \cite{banachLattice2021}, where the authors proved that under some assumptions on the initial condition, the solution to the auxiliary control problem $\eqref{eq:problemPauxiliary}$ is also a solution to the true one. 
Hence, one interesting question to analyze, since as we saw Equation \eqref{eq:optimalConsumptionAux} is not an admissible strategy for certain initial conditions, would be that of finding the optimal consumption strategy which preserves positivity, at least in one particular example where \eqref{eq:optimalConsumptionAux} is known to fail. 

We conclude by suggesting some strategies for future works,  extending the result shown in this paper. A next step forward in the full comprehension of Problem \eqref{eq:problemP} would be being able to treat also the case with arbitrary strictly positive initial conditions, thus extending our study to the strictly positive cone $E^{++}$. One possible strategy to move in this direction would be that of studying the problem in 
\begin{equation*}
E^{+}_{\epsilon} := \{K \in E\,|\, K(\cdot) \geq \epsilon\},
\end{equation*}
which would allow us to use the same methodology applied in this manuscript since $E^{+}_{\epsilon}$ is again a convex closed set. However, it is not yet clear how one could pass from the non-invariance in $E^{+}_{\epsilon}$ for every small $\epsilon >0$ to that in $E^{++}$.

Another possible approach is to consider Equation \eqref{eq:optimalProblem} with $\sigma = 0$, i.e. the case with no capital reallocation over space. In particular, in this case, it is possible to solve Equation \eqref{eq:optimalProblem} explicitly for every point of $\theta \in S^{1}$, since the problem now reduces to an infinite set of ODEs parameterized by $\theta$, which are only weakly coupled by the term in  Equation \eqref{eq:optimalConsumptionAux}. In particular, one can explicitly understand that some strictly positive initial conditions become negative, see Figure \ref{fig:sigma0} for a numerical example.
Then, one may prove that the solution to Equation $\eqref{eq:optimalProblem}$ exhibits continuity with respect to the parameter $\sigma$ in the value $0$, so that by taking $\sigma$ small enough the solution should be close to the one with no diffusion, see Figures \ref{fig:sigma1e-4} and \ref{fig:sigma1e-5}. Observe in particular how, when $\sigma = 10^{-5}$, the solution assumes negative values, as in the case with no diffusion. However, this continuity property must hold in a strong enough topology that allows the control of the solution in a pointwise manner. 

\begin{figure}[!htpb]
\centering
\begin{subfigure}[b]{0.57\textwidth}
\includegraphics[width=\textwidth]{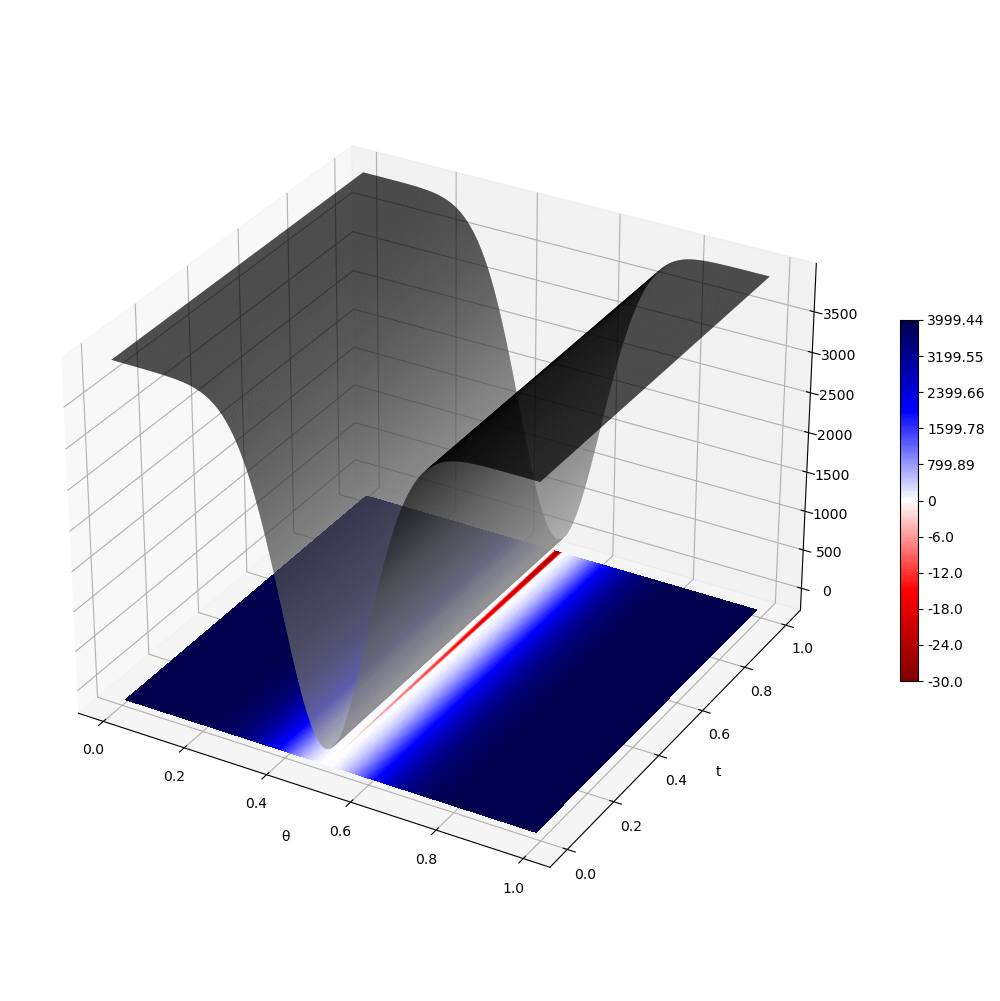}
\caption{$\sigma = 0$}
\label{fig:sigma0}
\end{subfigure}
\begin{subfigure}[b]{0.47\textwidth}
\centering
\includegraphics[width=\textwidth]{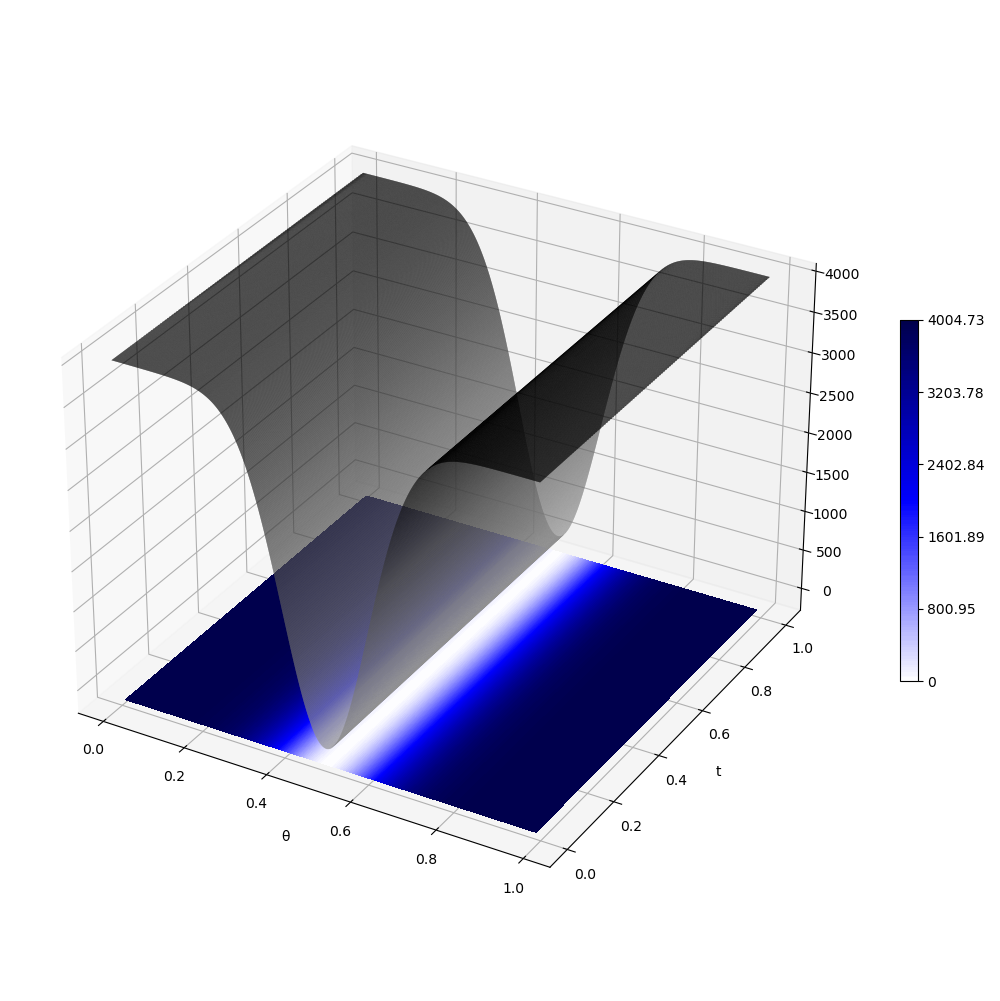}
\caption{$\sigma = 10^{-4}$}
\label{fig:sigma1e-4}
\end{subfigure}
\begin{subfigure}[b]{0.47\textwidth}
\centering
\includegraphics[width=\textwidth]{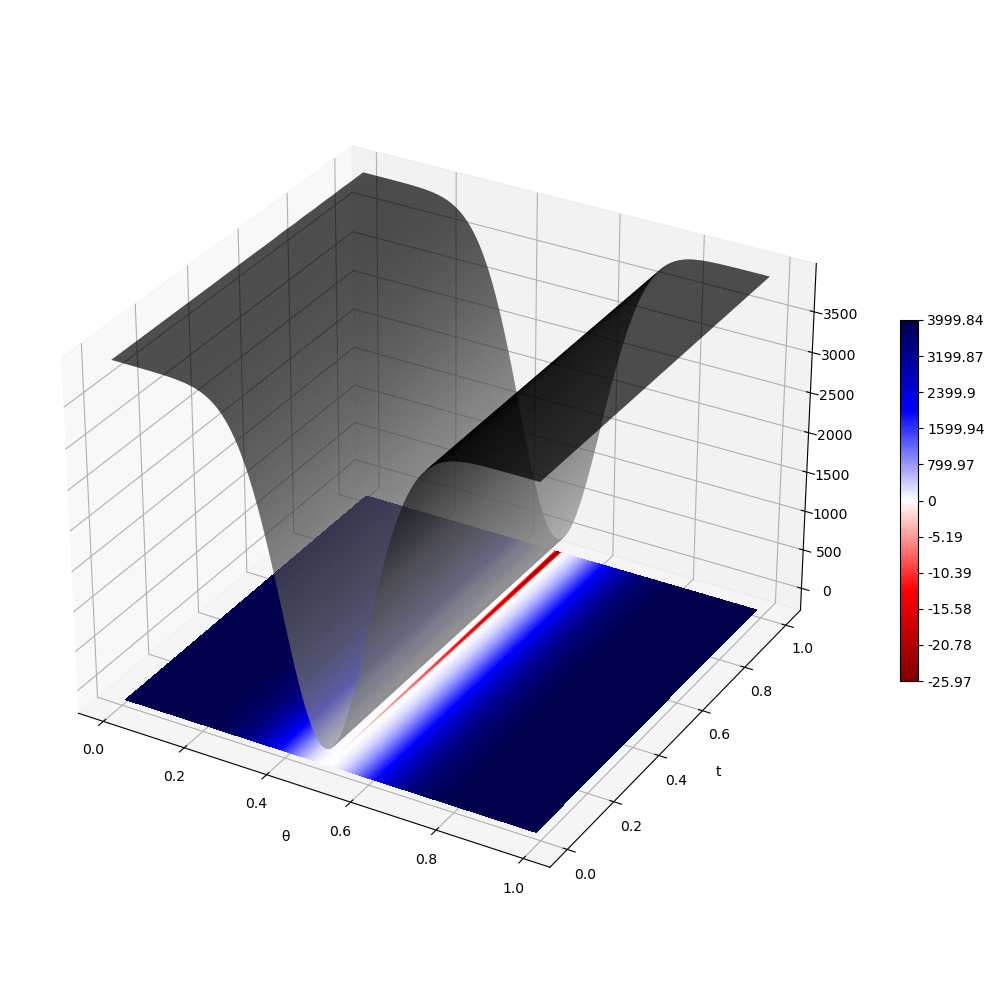}
\caption{$\sigma = 10^{-5}$}
\label{fig:sigma1e-5}
\end{subfigure}
\caption{Numerical results for different values of $\sigma$. The colormap below the surface plots shows positive (blue) and negative (red) values of the solution. }	
\label{fig:conclusion}	
\end{figure}

\noindent\textbf{Acknowledgements}: The author has been supported by the Italian Ministry of University and Research (MIUR), in the framework of PRIN project 2017FKHBA8 001 (The Time-Space Evolution of Economic Activities: Mathematical Models and Empirical Applications).\\
\noindent\textbf{Compliance with Ethical Standards}: The authors have no conflicts of interest to declare. The authors have no relevant financial or non-financial interests to disclose. This article does not contain any studies involving animals or humans performed by the author. 
\clearpage

\begin{thebibliography}{21}
  \providecommand{\natexlab}[1]{#1}
  \providecommand{\url}[1]{\texttt{#1}}
  \expandafter\ifx\csname urlstyle\endcsname\relax
    \providecommand{\doi}[1]{doi: #1}\else
    \providecommand{\doi}{doi: \begingroup \urlstyle{rm}\Url}\fi
  
  \bibitem[Arendt et~al.(1986)Arendt, Grabosch, Greiner, Moustakas, Nagel,
    Schlotterbeck, Groh, Lotz, and Neubrander]{arendt1986one}
  W.~Arendt, A.~Grabosch, G.~Greiner, U.~Moustakas, R.~Nagel, U.~Schlotterbeck,
    U.~Groh, H.~P. Lotz, and F.~Neubrander.
  \newblock \emph{One-parameter semigroups of positive operators}, volume 1184.
  \newblock Springer, Heidelberg, 1986.
  
  \bibitem[Bambi et~al.(2012)Bambi, Fabbri, and Gozzi]{bambi2012optimal}
  M.~Bambi, G.~Fabbri, and F.~Gozzi.
  \newblock Optimal policy and consumption smoothing effects in the time-to-build
    ak model.
  \newblock \emph{Economic Theory}, 50\penalty0 (3):\penalty0 635--669, 2012.
  
  \bibitem[Bambi et~al.(2017)Bambi, Di~Girolami, Federico, and
    Gozzi]{bambi2017generically}
  M.~Bambi, C.~Di~Girolami, S.~Federico, and F.~Gozzi.
  \newblock Generically distributed investments on flexible projects and
    endogenous growth.
  \newblock \emph{Economic Theory}, 63:\penalty0 521--558, 2017.
  
  \bibitem[Boucekkine et~al.(2013)Boucekkine, Camacho, and
    Fabbri]{boucekkine2013spatial}
  R.~Boucekkine, C.~Camacho, and G.~Fabbri.
  \newblock Spatial dynamics and convergence: The spatial ak model.
  \newblock \emph{Journal of Economic Theory}, 148\penalty0 (6):\penalty0
    2719--2736, 2013.
  
  \bibitem[Boucekkine et~al.(2019)Boucekkine, Fabbri, Federico, and
    Gozzi]{boucekkine2019growth}
  R.~Boucekkine, G.~Fabbri, S.~Federico, and F.~Gozzi.
  \newblock Growth and agglomeration in the heterogeneous space: a generalized ak
    approach.
  \newblock \emph{Journal of Economic Geography}, 19\penalty0 (6):\penalty0
    1287--1318, 2019.
  
  \bibitem[Boucekkine et~al.(2021)Boucekkine, Fabbri, Federico, and
    Gozzi]{boucekkine2021control}
  R.~Boucekkine, G.~Fabbri, S.~Federico, and F.~Gozzi.
  \newblock Control theory in infinite dimension for the optimal location of
    economic activity: The role of social welfare function.
  \newblock \emph{Pure and Applied Functional Analysis}, 6\penalty0 (5):\penalty0
    871--888, 2021.
  
  \bibitem[Calvia et~al.(2021)Calvia, Federico, and Gozzi]{banachLattice2021}
  A.~Calvia, S.~Federico, and F.~Gozzi.
  \newblock State constrained control problems in banach lattices and
    applications.
  \newblock \emph{SIAM Journal on Control and Optimization}, 59\penalty0
    (6):\penalty0 4481--4510, 2021.
  
  \bibitem[Calvia et~al.(2023)Calvia, Gozzi, Leocata, Papayiannis, Xepapadeas,
    and Yannacopoulos]{calvia2023optimal}
  A.~Calvia, F.~Gozzi, M.~Leocata, G.~I. Papayiannis, A.~Xepapadeas, and A.~N.
    Yannacopoulos.
  \newblock An optimal control problem with state constraints in a
    spatio-temporal economic growth model on networks.
  \newblock \emph{arXiv preprint arXiv:2304.11568}, 2023.
  
  \bibitem[Cannarsa and Di~Blasio(1995)]{cannarsa1995direct}
  P.~Cannarsa and G.~Di~Blasio.
  \newblock A direct approach to infinite-dimensional hamilton-jacobi equations
    and applications to convex control with state constraints.
  \newblock 1995.
  
  \bibitem[Cannarsa et~al.(1991)Cannarsa, Gozzi, and Soner]{cannarsa1991boundary}
  P.~Cannarsa, F.~Gozzi, and H.~M. Soner.
  \newblock A boundary-value problem for hamilton-jacobi equations in hilbert
    spaces.
  \newblock \emph{Applied Mathematics and Optimization}, 24\penalty0
    (1):\penalty0 197--220, 1991.
  
  \bibitem[Cannarsa et~al.(2018)Cannarsa, Da~Prato, and
    Frankowska]{cannarsa2018invariance}
  P.~Cannarsa, G.~Da~Prato, and H.~Frankowska.
  \newblock Invariance for quasi-dissipative systems in banach spaces.
  \newblock \emph{Journal of Mathematical Analysis and Applications},
    457\penalty0 (2):\penalty0 1173--1187, 2018.
  
  \bibitem[Capuzzo-Dolcetta and Lions(1990)]{capuzzo1990hamilton}
  I.~Capuzzo-Dolcetta and P.-L. Lions.
  \newblock Hamilton-jacobi equations with state constraints.
  \newblock \emph{Transactions of the American mathematical society},
    318\penalty0 (2):\penalty0 643--683, 1990.
  
  \bibitem[C{\^a}rja et~al.(2007)C{\^a}rja, Necula, and
    Vrabie]{carja2007viability}
  O.~C{\^a}rja, M.~Necula, and I.~I. Vrabie.
  \newblock \emph{Viability, invariance and applications}.
  \newblock Elsevier, Amsterdam, 2007.
  
  \bibitem[Fabbri and Gozzi(2008)]{fabbri2008solving}
  G.~Fabbri and F.~Gozzi.
  \newblock Solving optimal growth models with vintage capital: The dynamic
    programming approach.
  \newblock \emph{Journal of Economic Theory}, 143\penalty0 (1):\penalty0
    331--373, 2008.
  
  \bibitem[Faggian(2008)]{faggian2008hamilton}
  S.~Faggian.
  \newblock Hamilton--jacobi equations arising from boundary control problems
    with state constraints.
  \newblock \emph{SIAM journal on control and optimization}, 47\penalty0
    (4):\penalty0 2157--2178, 2008.
  
  \bibitem[Hartl et~al.(1995)Hartl, Sethi, and Vickson]{hartl1995survey}
  R.~F. Hartl, S.~P. Sethi, and R.~G. Vickson.
  \newblock A survey of the maximum principles for optimal control problems with
    state constraints.
  \newblock \emph{SIAM review}, 37\penalty0 (2):\penalty0 181--218, 1995.
  
  \bibitem[Martin(1973)]{martin1973differential}
  R.~Martin.
  \newblock Differential equations on closed subsets of a banach space.
  \newblock \emph{Transactions of the American Mathematical Society},
    179:\penalty0 399--414, 1973.
  
  \bibitem[Pavel(1977)]{pavel1977invariant}
  N.~Pavel.
  \newblock Invariant sets for a class of semi-linear equations of evolution.
  \newblock \emph{Nonlinear Analysis: Theory, Methods \& Applications},
    1\penalty0 (2):\penalty0 187--196, 1977.
  
  \bibitem[Pavel(1983)]{pavel1983semilinear}
  N.~H. Pavel.
  \newblock Semilinear equations with dissipative time-dependent domain
    perturbations.
  \newblock \emph{Israel Journal of Mathematics}, 46:\penalty0 103--122, 1983.
  
  \bibitem[Shuzhong(1989)]{shuzhong1989viability}
  S.~Shuzhong.
  \newblock Viability theorems for a class of differential-operator inclusions.
  \newblock \emph{Journal of differential equations}, 79\penalty0 (2):\penalty0
    232--257, 1989.
  
  \bibitem[Soner(1986)]{soner1986optimal}
  H.~M. Soner.
  \newblock Optimal control with state-space constraint i.
  \newblock \emph{SIAM Journal on Control and Optimization}, 24\penalty0
    (3):\penalty0 552--561, 1986.
  
  \end{thebibliography}

\end{document}